\documentclass[12pt]{article}

\usepackage{amsthm}
\usepackage[sumlimits]{amsmath}
\usepackage{amssymb}
\usepackage{algorithm}
\usepackage{algorithmic}
\usepackage[all]{xy}
\usepackage{wasysym}
\usepackage{fullpage}
\usepackage{graphicx}
\usepackage{url}
\usepackage[numbers,comma]{natbib}
\usepackage{authblk}
\usepackage{subfigure}
\usepackage[colorlinks,linkcolor=red,citecolor=blue,urlcolor=blue]{hyperref}
\usepackage{thm-restate}

\usepackage[disable]{todonotes}

\usepackage{times}
\usepackage{multitoc}
\usepackage[small]{titlesec}

\titlespacing{\section}{0pt}{2ex}{1ex}
\titlespacing{\subsection}{0pt}{1.5ex}{1ex}
\titlespacing{\subsubsection}{0pt}{1.5ex}{0.5ex}
\titlespacing{\paragraph}{%
  0pt}{
  0.5\baselineskip}{
  0.5em}%
\let\oldnormalsize\normalsize
\def\normalsize{\oldnormalsize%
\abovedisplayskip5pt plus2pt minus 2pt%
\belowdisplayskip5pt plus2pt minus 2pt}
\newif\ifdcc\dccfalse   

\newcommand{\arxivelink}{\cite{2012arXiv1211.0587V}}

\newcounter{appcount}
\newcommand{\appsec}[1]{
     \renewcommand{\theappcount}{\Alph{appcount}}
     \refstepcounter{appcount}
     \section*{Appendix \theappcount.\ #1}
     }

\newtheorem{defn}{Definition}
\newtheorem{lem}{Lemma}
\newtheorem{thm}{Theorem}

\newtheorem{proposition}{Proposition}

\newtheorem{corollary}{Corollary}
\theoremstyle{definition}

\def\cdbar{|}

\newcommand{\cX}{\mathcal{X}}
\newcommand{\cI}{\mathcal{I}}
\newcommand{\cM}{\mathcal{M}}
\newcommand{\cS}{\mathcal{S}}
\newcommand{\cC}{\mathcal{C}}

\newcommand{\cL}{\mathcal{L}}
\newcommand{\cP}{\mathcal{P}}
\newcommand{\Part}{\mathcal{P}}

\newcommand{\cT}{\mathcal{T}}
\newcommand{\cts}{\text{\sc cts}}

\newcommand{\ptw}{\text{\sc ptw}}

\newcommand{\lad}{\text{\sc lad}}
\newcommand{\deckt}{\text{\sc dec-kt}}
\newcommand{\ladg}{\text{\sc vcw}}

\newcommand{\kt}{\text{\sc kt}}

\newcommand{\seq}{x_{1:n}}

\newcommand{\nlogprior}[1]{\Gamma_d(#1)}

\def\ss#1{\begin{scriptsize}#1\end{scriptsize}}

\begin{document}

\def\name{}
\def\email{\small\hfill\sc}
\def\addr#1{\small\it #1}

\title{
\bf\LARGE\hrule height5pt \vskip 4mm
Partition Tree Weighting
\vskip 4mm \hrule height2pt
}

\renewcommand\Authsep{~~~~}
\renewcommand\Authand{~~~~}
\renewcommand\Authands{~~~~}
\setlength{\affilsep}{0.5em}

\author[$\ddagger\dagger$]{Joel Veness}
\author[$\dagger$]{Martha White}
\author[$\dagger$]{Michael Bowling}
\author[$\dagger$]{Andr\'{a}s Gy\"{o}rgy}
\affil[$\dagger$]{\small University of Alberta, Edmonton, Canada}
\affil[$\ddagger$]{\small DeepMind Technologies}

\date{}
\maketitle

\vspace{-3.5em}
\begin{abstract}
This paper introduces the Partition Tree Weighting technique, an efficient meta-algorithm for piecewise stationary sources.
The technique works by performing Bayesian model averaging over a large class of possible partitions of the data into locally stationary segments.
It uses a prior, closely related to the Context Tree Weighting technique of Willems, that is well suited to data compression applications.
Our technique can be applied to any coding distribution at an additional time and space cost only logarithmic in the sequence length. 
We provide a competitive analysis of the redundancy of our method, and explore its application in a variety of settings.
The order of the redundancy and the complexity of our algorithm matches those of the best competitors available in the literature,
and the new algorithm exhibits a superior complexity-performance trade-off in our experiments. 
\end{abstract}

\section{Introduction}

Coping with data generated from non-stationary sources is a fundamental problem in data compression.
Many real-world data sources drift or change suddenly, often violating the stationarity assumptions implicit in many models. 
Rather than modifying such models so that they robustly handle non-stationary data, one promising approach has been to instead design meta-algorithms that automatically generalize existing stationary models to various kinds of non-stationary settings.

A particularly well-studied kind of non-stationary source is the class of piecewise stationary sources.
Algorithms designed for this setting assume that the data generating source can be well modeled by a sequence of stationary sources.
This assumption is quite reasonable, as piecewise stationary sources have been shown \citep{adak1998time} to adequately handle various types of non-stationarity. 
Piecewise stationary sources have received considerable attention from researchers in information theory \citep{Willems96,willemsPSMS97,Shamir99lowcomplexity}, online learning \citep{herbster1998,Vov99,KoRo08,deRooij09,hazan2009efficient,gyorgy2011efficient}, time series \citep{adak1998time,davis2006structural}, and graphical models \citep{fearnhead2006exact,adams2007bayesian,angelosante2011sparse}.

An influential approach for piecewise stationary sources is the universal \emph{transition diagram} technique of \citet{Willems96} for the statistical data compression setting.
This technique performs Bayesian model averaging over all possible partitions of a sequence of data. 
Though powerful for modeling piecewise stationary sources, its quadratic time complexity makes it too computationally intensive for many applications. 
Since then, more efficient algorithms have been introduced that weight over restricted subclasses of partitions~\citep{willemsPSMS97,Shamir99lowcomplexity,hazan2009efficient,gyorgy2011efficient}. 
For example, Live and Die Coding \citep{willemsPSMS97} considers only $\log t$ partitions at any particular time $t$ by terminating selected partitions as time progresses, resulting in an $O(n \log n)$ algorithm for binary, piecewise stationary, memoryless sources with provable redundancy guarantees.\footnote{All logarithms in this paper are of base $2$.}
\citet{gyorgy2011efficient} recently extended and generalized these approaches into a parametrized online learning framework that can interpolate between many of the aforementioned weighting schemes for various loss functions. We also note that linear complexity algorithms exist for prediction in changing environments \cite{herbster1998, KoRo08, Zin03}, but these algorithms work with a restricted class of predictors and are not applicable directly to the data compression problem considered here. 

In this paper we introduce the Partition Tree Weighting (\ptw) technique, a computationally efficient meta-algorithm that also works by weighting over a large subset of possible partitions of the data into stationary segments.
Compared with previous work, the distinguishing feature of our approach is to use a prior, closely related to Context Tree Weighting \cite{ctw95}, that contains a strong bias towards partitions containing long runs of stationary data.
As we shall see later, this bias is particularly suited for data compression applications, while still allowing us to provide theoretical guarantees competitive with previous low complexity weighting approaches.

\section{Background}

We begin with some terminology for sequential, probabilistic data generating sources.
An alphabet is a finite, non-empty set of symbols, which we will denote by $\cX$. 
A string $x_1x_2 \ldots x_n \in \cX^n$ of length $n$ is denoted by $x_{1:n}$.
The prefix $x_{1:j}$ of $x_{1:n}$, $j\leq n$, is denoted by $x_{\leq j}$ or $x_{< j+1}$.
The empty string is denoted by $\epsilon$.
Our notation also generalises to out of bounds indices; that is, given a string $x_{1:n}$ and an integer $m > n$, we define $x_{1:m} := x_{1:n}$ and $x_{m:n}:=\epsilon$.
The concatenation of two strings $s$ and $r$ is denoted by $sr$.

\paragraph{Probabilistic Data Generating Sources.}

A probabilistic data generating source $\rho$ is defined by a sequence of probability mass functions $\rho_n : \cX^n \to [0,1]$, for all $n\in\mathbb{N}$, satisfying the compatibility constraint that 
$\rho_n(x_{1:n}) = \sum_{y\in\cX} \rho_{n+1}(x_{1:n}y)$
for all $x_{1:n} \in \cX^n$, with base case $\rho_0(\epsilon) = 1$.
From here onwards, whenever the meaning is clear from the argument to $\rho$, the subscripts on $\rho$ will be dropped.
Under this definition, the conditional probability of a symbol $x_n$ given previous data $x_{<n}$ is defined as $\rho(x_n | x_{<n}) := \rho(x_{1:n}) / \rho(x_{<n})$ provided $\rho(x_{<n}) > 0$, with the familiar chain rules $\rho(x_{1:n}) = \prod_{i=1}^n \rho(x_i | x_{<i})$ and $\rho(x_{i:j} \cdbar x_{<i}) = \prod_{k=i}^j \rho(x_k | x_{<k})$ now following.

\paragraph{Temporal Partitions.}
Now we introduce some notation to describe temporal partitions.
A segment is a tuple $(a,b) \in \mathbb{N}\times\mathbb{N}$ with $a \leq b$.
A segment $(a,b)$ is said to overlap with another segment $(c,d)$ if there exists an $i \in \mathbb{N}$ such that $a \leq i \leq b$ and $c \leq i \leq d$.
A temporal partition $\cP$ of a set of time indices $S = \{ 1,2,\dots\,n \}$, for some $n\in\mathbb{N}$, is a set of non-overlapping segments such that for all $x\in\cS$, there exists a segment $(a,b) \in \cP$ such that $a \leq x \leq b$.
We also use the overloaded notation $\cP(a,b) := \{ (c,d) \in \cP \;:\; a \leq c \leq d \leq b \}$.
Finally, $\cT_n$ will be used to denote the set of all possible temporal partitions of $\{ 1, 2, \dots, n \}$.

\paragraph{Piecewise Stationary Sources.}
We can now define a piecewise stationary data generating source $\mu$ in terms of a partition $\cP = \left\{ (a_1,b_1), (a_2,b_2), \dots \right \}$ and a set of probabilistic data generating sources $\{ \mu^1, \mu^2, \dots \},$ such that for all $n \in \mathbb{N}$, for all $x_{1:n} \in \cX^n$,
\begin{equation*}
\mu(x_{1:n}) := 
\prod_{(a,b)\in\cP_n} \mu^{f(a)}(x_{a:b}),
\end{equation*}
where $\cP_n := \left \{ (a_i, b_i) \in \cP \,:\, a_i \leq n \right \}$ and $f(i)$ returns the index of the time segment containing $i$, that is, it gives a value $k \in \mathbb{N}$ such that both $(a_k, b_k) \in \cP$ and $a_k \leq i \leq b_k$.

\paragraph{Redundancy.}

The ideal code length given by a probabilistic model (or probability assignment) $\rho$ on a data sequence $\seq  \in \cX^n$ is given by $-\log \rho(\seq)$, and the redundancy of $\rho$, with respect to a probabilistic data generating source $\mu$, is defined as $\log \mu(\seq) - \log \rho(\seq)$. 
This quantity corresponds to the amount of extra bits we would need to transmit $x_{1:n}$ using an optimal code designed for $\rho$ (assuming the ideal code length) compared to using an optimal code designed for the data generating source $\mu$.

\section{Partition Tree Weighting}

Almost all existing prediction algorithms designed for changing source statistics are based on the transition diagram technique of \citet{Willems96}.
This technique performs exact Bayesian model averaging over the set of temporal partitions, or more precisely, the method averages over all coding distributions formed by employing a particular base model $\rho$ on all segments of every possible partition. 
Averaging over all temporal partitions (also known as transition paths) results in an algorithm of $O(n^2)$ complexity.
Several reduced complexity methods were proposed in the literature that average over a significantly smaller
set of temporal partitions \cite{willemsPSMS97,hazan2009efficient,gyorgy2011efficient}: the reduced number of partitions allows for the computational complexity to be pushed down to $O(n\log n)$, while still being sufficiently rich to guarantee almost optimal redundancy behavior (typically $O(\log n)$ times larger than the optimum, for lossless data compression). 
In this paper we propose another member of this family of methods. 
Our reduced set of temporal partitions, as well as the corresponding mixture weights, are obtained from the Context Tree Weighting (CTW) algorithm \cite{ctw95}, which results in similar theoretical guarantees as the other methods, but shows superior performance in all of our experiments. 
The method, called Partition Tree Weighting, heavily utilizes the computational advantages offered by the CTW algorithm.

We now derive the Partition Tree Weighting (\ptw) technique. 
As \ptw\ is a meta-algorithm, it takes as input a base model which we denote by $\rho$ from here onwards.
This base model determines what kind of data generating sources can be processed.

\subsection{Model Class}

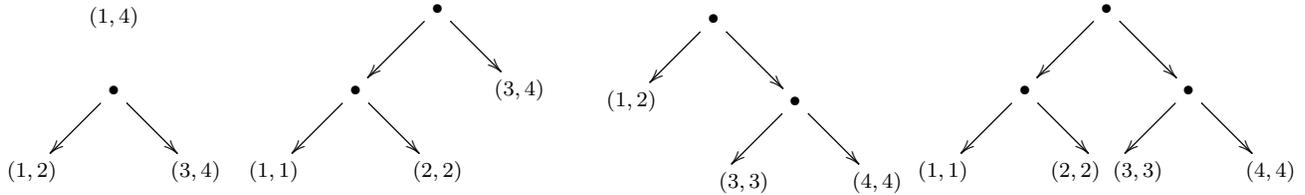
\begin{figure}
\scriptsize
\centerline{
\begin{tabular}[b]{c}
\xymatrix @ur {
(1,4)
} \\[2.5em]
\xymatrix @ur {
(1,2) & \bullet \ar[l] \ar[d] \\
& (3,4) 
}
\end{tabular}
\xymatrix @ur {
(1,1)   & \bullet \ar[l] \ar[d] & \bullet \ar[l] \ar[d] \\
& (2,2) & (3,4) \\
}
\xymatrix @ur {
& (1, 2) & \bullet \ar[l] \ar[d] \\
& \quad\quad\;\;\; (3,3) & \bullet \ar[l] \ar[d] \\
&  & (4,4)
}
\xymatrix @ur {
(1,1)   &\bullet \ar[l] \ar[d] & \bullet \ar[l] \ar[d] \\
& (2,2) \;\; (3,3) & \bullet \ar[l] \ar[d] \\
&  & (4,4)
}
}
\vspace{-8.0em}
\caption{The set $\cC_2$ represented as a collection of partition trees.}\label{fig:partition_tree}
\vspace{-0.3cm}
\end{figure}

We begin by defining the class of binary temporal partitions.
Although more restrictive than the class of all possible temporal partitions, binary temporal partitions possess important computational advantages that we will later exploit. 

\begin{defn}\label{def:binary_temporal_partitions}
Given a depth parameter $d \in \mathbb{N}$ and a time $t \in \mathbb{N}$, the set $\cC_d(t)$ of all binary temporal partitions from $t$ is recursively defined by
\begin{equation*}
\cC_d(t) := \bigl\{ \{ (t,t+2^{d}-1) \} \bigr\} \cup \left\{ \cS_1 \cup \cS_2 : \cS_1 \in \cC_{d-1} \left(t \right), \cS_2 \in \cC_{d-1} \left(t + 2^{d-1} \right)  \right\},
\end{equation*}
with $\cC_0(t) := \bigl\{  \{ (t,t) \} \bigr\}$.
Furthermore, we define $\cC_d := \cC_d(1)$. 
\end{defn}
For example, $\cC_2 =$ $\bigl\{$ $\{ (1,4) \}$, $\{ (1,2), (3,4) \}$, $\{ (1,1), (2,2), (3,4) \}$, $\{ (1,2), (3,3), (4,4) \}$, $\{ (1,1), (2,2), (3,3), (4,4) \}$ $\bigr\}$. 
Each partition can be naturally mapped onto a tree structure which we will call a partition tree.
Figure~\ref{fig:partition_tree} shows the collection of partition trees represented by $\cC_2$.
Notice that the number of binary temporal partitions $|\cC_d|$ grows roughly double exponentially in $d$.
For example, $|\cC_0| = 1$, $|\cC_1| = 2$, $|\cC_2| = 5$, $|\cC_3| = 26$, $|\cC_4|=677$, $|\cC_5| = 458330$, which means that some ingenuity will be required to weight over all $\cC_d$ efficiently.

\subsection{Coding Distribution}

\newcommand{\ptwdefinner}{2^{-\Gamma_d(\cP)} \hspace{-0.45em} \prod_{(a,b) \in \cP} \hspace{-0.35em} \rho(x_{a:b})}
\newcommand{\ptwdef}{\sum_{\cP \in \cC_d} \ptwdefinner}
We now consider a particular weighting over $\cC_d$ that has both a bias towards simple partitions and efficient computational properties.
Given a data sequence $x_{1:n}$, we define
\begin{align}\label{eq:ptw}
\ptw_d(x_{1:n}) := \ptwdef, 
\end{align}
where $\Gamma_d(\cP)$ gives the number of nodes in the partition tree associated with $\cP$ that have a depth less than $d$.
This prior weighting is identical to how the Context Tree Weighting method \citep{ctw95} weights over tree structures, and is an application of the general technique used by the class of Tree Experts described in Section $5.3$ of \cite{pred_learning_games}.
It is a valid prior, as one can show $\sum_{\cP \in \cC_d} 2^{-\Gamma_d(\cP)}=1$ for all $d \in \mathbb{N}$.
Note that Algorithm~\ref{alg:ptw} is a special case, using a prior $2^{-\Gamma_d(\cP)}$ for $\cP \in \cC_d$ and $0$ otherwise, of the class of general algorithms discussed in \cite{gyorgy2011efficient}.
As such, the main contribution of this paper is the introduction of this specific prior.
A direct computation of Equation \ref{eq:ptw} is clearly intractable.
Instead, an efficient approach can be obtained by noting that Equation \ref{eq:ptw} can be recursively decomposed.

\begin{restatable}{lem}{lemmaPtwLemma}\label{lem:ptw_lemma}
For any depth $d \in \mathbb{N}$, given a sequence of data $x_{1:n} \in \cX^n$ satisfying $n \leq 2^d$, 
\begin{equation}\label{eq:ptw_lemma}
\ptw_d(x_{1:n}) = \frac{1}{2} \rho(x_{1:n}) + \frac{1}{2} \ptw_{d-1} \left( x_{1:k} \right) \ptw_{d-1}\left(x_{k+1:n} \right),
\end{equation}
where $k=2^{d-1}$.
\end{restatable}
\begin{proof}
\ifdcc
\else
A straightforward adaptation of \citep[Lemma~2]{ctw95}, see Appendix~\ref{sec:proofs}.
\fi
\end{proof}

\subsection{Algorithm}

\begin{figure}
\hspace{0em}
\centerline{
\xymatrix @ur {
(1,1)   & \mathbf{(1,2)} \ar[l]_0 \ar[d]^{\mathbf{1}} & \bf{(1,4)} \ar[l]_{\mathbf{0}} \ar[d]^1 \\
& \mathbf{(2,2)} \;\; (3,3) & (3,4) \ar[l]_0 \ar[d]^1 \\
&  & (4,4)
}
\hspace{0.25em}
\xymatrix @ur {
(1,1)   & (1,2) \ar[l]_0 \ar[d]^1 & \mathbf{(1,4)} \ar[l]_0 \ar[d]^{\mathbf{1}} \\
& (2,2) \;\; \mathbf{(3,3)} & \mathbf{(3,4)} \ar[l]_0 \ar[d]^1 \\
&  & (4,4)
}
}
\vspace{-6.5em}
\caption{Partitions updated at $t=2$ (left) and $t=3$ (right) in a depth-2 partition tree.}\label{fig:partition_tree_updates}
\end{figure}
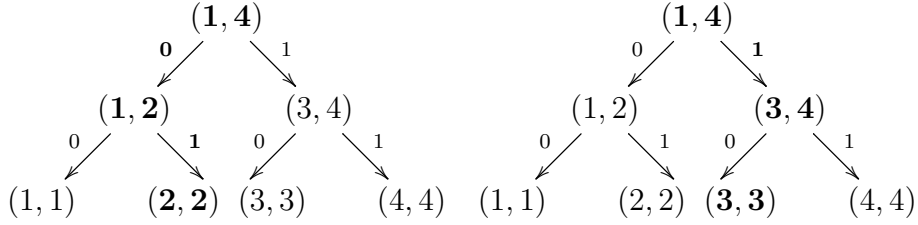

Lemma~\ref{lem:ptw_lemma} allows us to compute $\ptw_d(x_{1:n})$ in a bottom up fashion.
This leads to an algorithm that runs in $O(nd)$ time \emph{and} space by maintaining a context tree data structure in memory.
One of our main contributions is to further reduce the space overhead of \ptw\ to $O(d)$ by exploiting the regular access pattern to this data structure.
To give some intuition, Figure \ref{fig:partition_tree_updates} shows in bold the nodes in a context tree data structure that would need to be updated at times $t=2$ and $t=\nobreak3$.
The key observation is that because our access patterns are performing a kind of depth first traversal of the context tree, the needed statistics can be summarized in a stack of size $d$.
This has important practical significance, as the performance of many interesting base models will depend on how much memory is available.

Algorithm \ref{alg:ptw} describes our $O(n d)$ time and $O(d)$ space technique for computing $\ptw_d(x_{1:n})$. 
It uses a routine, $\text{{\sc mscb}}_d(t)$, that returns the most significant changed context bit; that is, for $t > 1$, this is the number of bits to the left of the most significant location at which the $d$-bit binary representations of $t-1$ and $t-2$ differ, with $\text{\sc mscb}_d(1) := 0$ for all $d\in \mathbb{N}$.
For example, for $d=5$, we have $\text{\sc mscb}_5(4) = 4$ and $\text{\sc mscb}_5(7) = 3$.
Since $d = \lceil \log n \rceil$, the algorithm effectively runs in $O(n \log n)$ time and $O(\log n)$ space.
Furthermore, Algorithm \ref{alg:ptw} can be modified to run incrementally, as $\ptw_d(x_{1:n})$ can be computed from $\ptw_d(x_{<n})$ in $O(d)$ time provided the intermediate buffers $b_i$, $w_i$, $r_i$ for $0 \leq i \leq d$ are kept in memory.

\algsetup{indent=2em}
\begin{algorithm}[t!]
\caption{{\sc Partition Tree Weighting - $\ptw_{d}(x_{1:n})$}\label{alg:ptw}}
\begin{algorithmic}[1]
\REQUIRE A depth parameter $d \in \mathbb{N}$
\REQUIRE A data sequence $x_{1:n} \in \cX^n$ satisfying $n \leq 2^d$
\REQUIRE A base probabilistic model $\rho$
\medskip
\STATE $b_j \leftarrow 1, w_j \leftarrow 1, r_j \leftarrow 1, \text{~for~} 0 \leq j \leq d$
\medskip
\FOR{$t=1$ to $n$}
	  \STATE $i \leftarrow \text{\sc mscb}_d(t)$
	  \STATE $b_{i} \leftarrow w_{i+1}$ 
		\FOR{$j=i+1$ to $d$}
		    \STATE $r_j \leftarrow t$
		\ENDFOR 	  
    \STATE $w_d \leftarrow \rho(x_{r_d:t})$
    \FOR{$i=d-1$ to $0$}
        \STATE $w_{i} \leftarrow \tfrac{1}{2} \rho(x_{r_i:t}) + \tfrac{1}{2} w_{i+1} b_{i} $
    \ENDFOR
\ENDFOR
\medskip
\RETURN $w_0$
\end{algorithmic}
\end{algorithm}

\subsection{Theoretical Properties }

We now provide a theoretical analysis of the Partition Tree Weighting method.
\ifdcc
\else
Our first result shows that using \ptw\ with a base model $\rho$ is almost as good as using $\rho$ with any partition in the class $\cC_d$.
\begin{proposition}\label{prop:ptw_weighting_bound}
For all $n\in\mathbb{N}$, where $d = \lceil \log n \rceil $, for all $x_{1:n} \in \cX$, for all $\cP\in\cC_d$, we have
\begin{equation*}
-\log \ptw_d(x_{1:n}) \leq ~~ \Gamma_d(\cP) ~ + \hspace{-0.4em} \sum_{(a,b) \in \cP} -\log \rho(x_{a:b}).
\end{equation*}
\vspace{-0.5em}
\begin{proof}
We have that
\vspace{-1em}
\begin{equation*}
-\log \ptw_d(x_{1:n}) = -\log \left( \ptwdef \right)
\leq \Gamma_d(\cP) -\log \hspace{-0.8em} \prod_{(a,b) \in \cP} \hspace{-0.7em} \rho(x_{a:b}),
\end{equation*}
where $\cP$ is an arbitrary partition in $\cC_d$.
Rewriting the last term completes the proof.
\end{proof}
\end{proposition}
\fi

The next result shows that there always exists a binary temporal partition (i.e., a partition in $\cC_d$) which is in some sense close to any particular temporal partition.
To make this more precise, we introduce some more terminology.
First we define $C(\cP) := \{ a \}_{(a,b) \in \cP} \setminus \{ 1 \}$, which is the set of time indices where an existing segment ends and a new segment begins in partition $\cP$.
Now, if $C(\cP) \subseteq C(\cP')$, we say $\cP'$ is a refinement of partition $\cP$.
In other words, $\cP'$ is a refinement of $\cP$ if $\cP'$ always starts a new segment whenever $\cP$ does.
With a slight abuse of notation, we will also use $C(\cP)$ to denote the partition $\cP$.
\begin{lem}\label{lem:split_overhead}
For all $n \in \mathbb{N}$, for any temporal partition $\cP \in \cT_n$, with $d = \lceil \log n \rceil$, there exists a binary temporal partition $\cP' \in \cC_d$ such that $\cP'$ is a refinement of $\cP$ and $|\cP'| \leq |\cP| ( \lceil \log n \rceil + 1) $.
\begin{proof}
We prove this via construction.
Consider a binary tree $T_i$ with $1 \leq i \leq n$ formed from the following recursive procedure:
1. Set $a=1$, $b=2^d$, add the node $(a,b)$ to the tree.
2. If $a=i$ then stop;
otherwise, add $(a,\left \lfloor \tfrac{b-a}{2} \right\rfloor )$, $(a + \left \lfloor \tfrac{b-a}{2} \right\rfloor +1, b)$ as children to node $(a,b)$, and then set $(a,b)$ to the newly added child containing $i$ and goto step 2.

Next define $\cL(T_i)$ and $\cI(T_i)$ to be the set of leaf and internal nodes of $T_i$ respectively. 
Notice that $|\cL(T_i)| \leq d+1$ and that $\cL(T_i) \in \cC_d$.
Now, consider the set 
$$
\cP' := \left \{ (a,b) \in \bigcup\limits_{i \in C(\cP)} \cL(T_i) \,:\, (a,b) \notin \bigcup\limits_{i \in C(\cP)} \cI(T_i) \right \}.
$$
It is easy to verify that $\cP' \in \cC_d$ and that $C(\cP) \subseteq C(\cP')$.
The proof is concluded by noticing that 
\begin{equation*}
|\cP'| \leq \biggl| \bigcup\limits_{i \in C(\cP)} \cL(T_i) \biggr| \leq \sum_{i \in C(\cP)} |\cL(T_i)| 
\leq |C(\cP) | (d+1)
\leq |\cP| \left( \left\lceil \log n \right\rceil + 1\right).\qedhere
\end{equation*}
\end{proof}
\end{lem}

Next we show that if we have a redundancy bound for the base model $\rho$ that holds for any finite sequence of data generated by some class of bounded memory data generating sources, we can automatically derive a redundancy bound when using \ptw\ on the piecewise stationary extension of that same class.

\begin{thm}\label{thm:ptw_redundancy}
For all $n\in\mathbb{N}$, using \ptw\ with $d = \lceil \log n \rceil$ and a base model $\rho$ whose redundancy is upper bounded by a non-negative, monotonically non-decreasing, concave function $g~:~\mathbb{N} \to \mathbb{R}$ with $g(0)=0$ on some class ${\cal G}$ of bounded memory data generating sources, the redundancy
\begin{equation*}
\log \mu(\seq) -\log \ptw_d(\seq) \leq \nlogprior{\Part'} + |\cP| \, g\left( \left\lceil \frac{n}{|\cP| ( \lceil \log n \rceil + 1)} \right\rceil \right) ( \lceil \log n \rceil + 1), 
\end{equation*}
where $\mu$ is a piecewise stationary data generating source, and the data in each of the stationary regions $\cP \in \cT_n$ is distributed according to some source in ${\cal G}$.
Furthermore, $\Gamma_d(\cP')$ can be upper bounded independently of $d$ by $2 |\cP| \left( \left\lceil \log n \right\rceil + 1 \right)$.
\ifdcc
\begin{proof}
Combine Lemma 1 in \cite{gyorgy2011efficient} with our Lemma \ref{lem:split_overhead} and the properties of $g$, as per the proof of Theorem 2 in \cite{gyorgy2011efficient}.
See \arxivelink\ for more details, including the proof of the upper bound on $\Gamma_d(\cP')$.
\end{proof}
\end{thm}
\else
\begin{proof}
From Lemma ~\ref{lem:split_overhead}, we know that there exists a partition $\cP' \in \cC_d$ that is a refinement of $\cP$ containing at most $|\cP| ( \lceil \log n \rceil + 1)$ segments.
Applying Proposition \ref{prop:ptw_weighting_bound} to $\cP'$ and using our redundancy bound $g$ gives
\begin{align*}
&\log \mu(\seq) -\log \ptw_d(\seq) \le \log \mu(\seq) + \nlogprior{\cP'} - \sum_{(a,b) \in \cP'}  \log \rho(x_{a:b})\\
&\hspace{0.3cm}= \nlogprior{\cP'} + \sum_{(a,b)  \in \cP} \log \mu^{f(a)}(x_{a:b}) - \sum_{(a,b) \in \cP'} \log \rho(x_{a:b})\\
&\hspace{0.3cm}= \nlogprior{\cP'} - \sum_{(a,b) \in \Part'} \log \rho(x_{a:b}) + \sum_{(a,b)  \in \cP} \sum_{(c,d)  \in \cP'(a,b)} \log \mu^{f(a)}(x_{c:d} \cdbar x_{a:c-1})\\
&\hspace{0.3cm}  
\leq  \nlogprior{\Part'}  + \sum_{(a,b) \in \Part'} g(b-a+1).
\end{align*}
Since  by definition $g$ is concave, Jensen's inequality implies
$\sum_{(a,b) \in \Part'} g(b-a+1) \le |\Part'| g(n/|\Part'|)$. Furthermore, the conditions on $g$ also imply that $a g(b/a)$ is a nondecreasing function of $a>0$ for any fixed $b>0$, and so Lemma~\ref{lem:split_overhead} implies
\begin{equation*}
 \sum_{(a,b) \in \Part'} g(b-a+1) \le
|\cP| \, g\left( \left\lceil \frac{n}{|\cP| ( \lceil \log n \rceil + 1) } \right\rceil \right) ( \lceil \log n \rceil + 1), 
\end{equation*}
which completes the main part of the proof.
The term $\nlogprior{\Part'}$ is proportional to the size of the partition tree needed to capture the locations where the data source changes.
Since at least half of the nodes of a binary tree are leaves, and the number of leaf nodes is the same as the number of segments in a partition, for any $d$ we have
$\Gamma_d(\cP') \leq 2 |\cP'| \leq 2 |\cP| \left( \left\lceil \log n \right\rceil + 1 \right)$.
\end{proof}
\end{thm}

We also remark that Theorem \ref{thm:ptw_redundancy} can be obtained by combining Lemma 1 in \cite{gyorgy2011efficient} with our Lemma \ref{lem:split_overhead} and the properties of $g$, as per the proof of Theorem 2 in \cite{gyorgy2011efficient}. However, to be self-contained, we decided to include the above short proof of the theorem. 
\fi

\paragraph{Removing the dependence on $n$ and $d$.}

Our previous results required choosing a depth $d$ in advance such that $d = \lceil \log n \rceil$.
This restriction can be lifted by using the modified coding distribution given by
$
\ptw(x_{1:n}) := \prod_{i=1}^n \ptw_{\lceil \log i \rceil}(x_{i} \cdbar x_{<i}).
$
The next result justifies this choice.
\begin{restatable}{thm}{propStronglyOnline}\label{thm:strongly_online}
For all $n \in \mathbb{N}$, for all $x_{1:n} \in \cX^n$, we have that
\begin{equation*}
-\log \ptw(x_{1:n}) \leq -\log \ptw_d(x_{1:n}) + \lceil \log n \rceil (\log 3 -1),
\end{equation*}
where $d=\lceil \log n \rceil$.
\end{restatable}
Thus, the overhead due to not knowing $n$ in advance is $O(\log n)$. The proof of this result, given in 
\ifdcc\else Appendix~\ref{sec:proofs}, \fi is based on the fact that for any $t,k \in \mathbb{N}$ satisfying $1 \le t\le 2^k$, we have
\begin{equation}
\label{eq:2/3penalty}
\tfrac{2}{3} \ptw_{k+1}(x_{1:t}) \leq \ptw_{k}(x_{1:t}),
\end{equation} 
which implies that each time the depth of the tree is increased, the algorithm suffers at most 
an extra $\log(3/2)$ penalty.

Algorithm \ref{alg:ptw} can be straightforwardly modified to compute $\ptw(x_{1:n})$ using the same amount of resources as needed for $\ptw_d(x_{1:n})$.
The main idea is to increase the size of the stack by one and copy over the relevant statistics whenever a power of two boundary is crossed.
Alternatively, one could simply pick a sufficiently large value of $d$ in advance.
This is also justified, as, based on Equation~\ref{eq:2/3penalty}, the penalty for using an unnecessarily large $k > d$ can be bounded by
\begin{equation*}
-\log \ptw_k(x_{1:n}) \leq -\log \ptw_d(x_{1:n}) + (k - d) \log \tfrac{3}{2}.
\end{equation*}

\section{Applications}

We now explore the performance of Partition Tree Weighting in a variety of settings.

\paragraph{Binary, Memoryless, Piecewise Stationary Sources.}
First we investigate using the well known \kt\ estimator \citep{krichevsky1981pue} as a base model for \ptw.
We begin with a brief overview of the \kt\ estimator. 
Consider a sequence $x_{1:n} \in \{0,1\}^n$ generated by successive Bernoulli trials.
If $a$ and $b$ denote the number of zeroes and ones in $x_{1:n}$ respectively, and $\theta \in [0,1]$ denotes the probability of observing a 1 on any given trial, then $\Pr(x_{1:n}\,|\,\theta) = \theta^b (1-\theta)^a$.
One way to construct a distribution over $x_{1:n}$, in the case where $\theta$ is unknown, is to weight over the possible values of $\theta$.
The KT-estimator uses the weighting $w(\theta) := \text{Beta($\tfrac{1}{2}$,$\tfrac{1}{2}$)} = \pi^{-1} \theta^{-1/2}(1-\theta)^{-1/2}$, which gives the coding distribution
$
\kt( x_{1:n} ) 
  := \int_0^1 \theta^b(1-\theta)^a w(\theta)\,d\theta
$
.
This quantity can be efficiently computed online by maintaining the $a$ and $b$ counts incrementally and using the chain rule, that is,
$\Pr(x_{n+1}=1|x_{1:n})=1-\Pr(x_{n+1}=0|x_{1:n}) = (b+1/2)/(n+1)$.
Furthermore, the parameter redundancy can be bounded uniformly; restating a result from \cite{ctw95}, one can show that for all $n \in \mathbb{N}$, for all $x_{1:n} \in \cX^n$, for all $\theta \in [0,1]$,
\begin{equation}\label{eq:kt_parameter_redun}
\log \theta^b (1-\theta)^a - \log \kt(x_{1:n}) \leq \tfrac{1}{2}\log(n) + 1.
\end{equation}

We now analyze the performance of the \kt\ estimator when used in combination with \ptw.
Our next result follows immediately from Theorem \ref{thm:ptw_redundancy} and Equation \ref{eq:kt_parameter_redun}.
\begin{corollary}\label{thm:ptw_kt}
For all $n \in \mathbb{N}$, for all $x_{1:n} \in \{0,1\}^n$, if $\mu$ is a piecewise stationary source segmented according to any partition $\cP \in \cT_n$, with the data in segment $i$ being generated by i.i.d. $Bernoulli(\theta_i)$ trials with $\theta_i \in [0,1]$ for $1 \leq i \leq |\cP|$, the redundancy of the \ptw-\kt\ algorithm, obtained by setting $d = \lceil \log n \rceil$ and using the \kt\ estimator as a base model, is upper bounded by
\begin{equation*}
\nlogprior{\Part'} + \frac{|\cP|}{2} \log \left\lceil \frac{n}{|\cP| ( \lceil \log n \rceil + 1)} \right\rceil ( \lceil \log n \rceil + 1) + |\cP| ( \lceil \log n \rceil + 1).
\end{equation*}
\end{corollary}
\noindent Corollary \ref{thm:ptw_kt} shows that the redundancy behavior of \ptw-\kt\ is $O\left(|\cP| (\log n)^2 \right)$.
Thus we expect this technique to perform well if the number of stationary segments is small relative to the length of the data.
This bound also has the same asymptotic order as previous \cite{willemsPSMS97, hazan2009efficient, gyorgy2011efficient} low complexity techniques.

Next we present some experiments with \ptw-\kt\ on synthetic, piecewise stationary data. 
We compare against techniques from information theory and online learning, including (i) Live and Die Coding (\lad) \citep{willemsPSMS97}, (ii) the variable complexity, exponential weighted averaging (\ladg) method of \cite{gyorgy2011efficient}, and (iii) the \deckt\ estimator \citep{oneil12,cts_dcc}, a heuristic variant of the \kt\ estimator that exponentially decays the $a$,$b$ counts to better handle non-stationary sources.
Figure \ref{fig:manysplits} illustrates the redundancy of each method as the number of change points increases.
To mitigate the effects of any particular choice of change points or $\theta_i$, we report results averaged over 50 runs, with each run using a uniformly random set of change points locations and $\theta_i$.  
95\% confidence intervals are shown on the graphs. 
\ptw\ performs noticeably better than all methods when the number of segments is small.
When the number of segments gets high, both \ptw\ and \ladg\ outperform all other techniques, with \ladg\ being slightly better once the number of change points is sufficiently large.
The $\ladg(g)$ technique interpolates between a weighting scheme very close to the linear method in \cite{Willems96} and \lad; however, the complexity of this algorithm when applied to the \kt\ estimator is $O(g n \log n)$, so with $g = 5$, the runtime is already 5 times larger than \ptw.

\begin{figure*}[t]
  \centering
   \subfigure[$n = 8192$ over splits $\{0,1, 2, \ldots, 20\}$]{\includegraphics[width=0.495\textwidth]{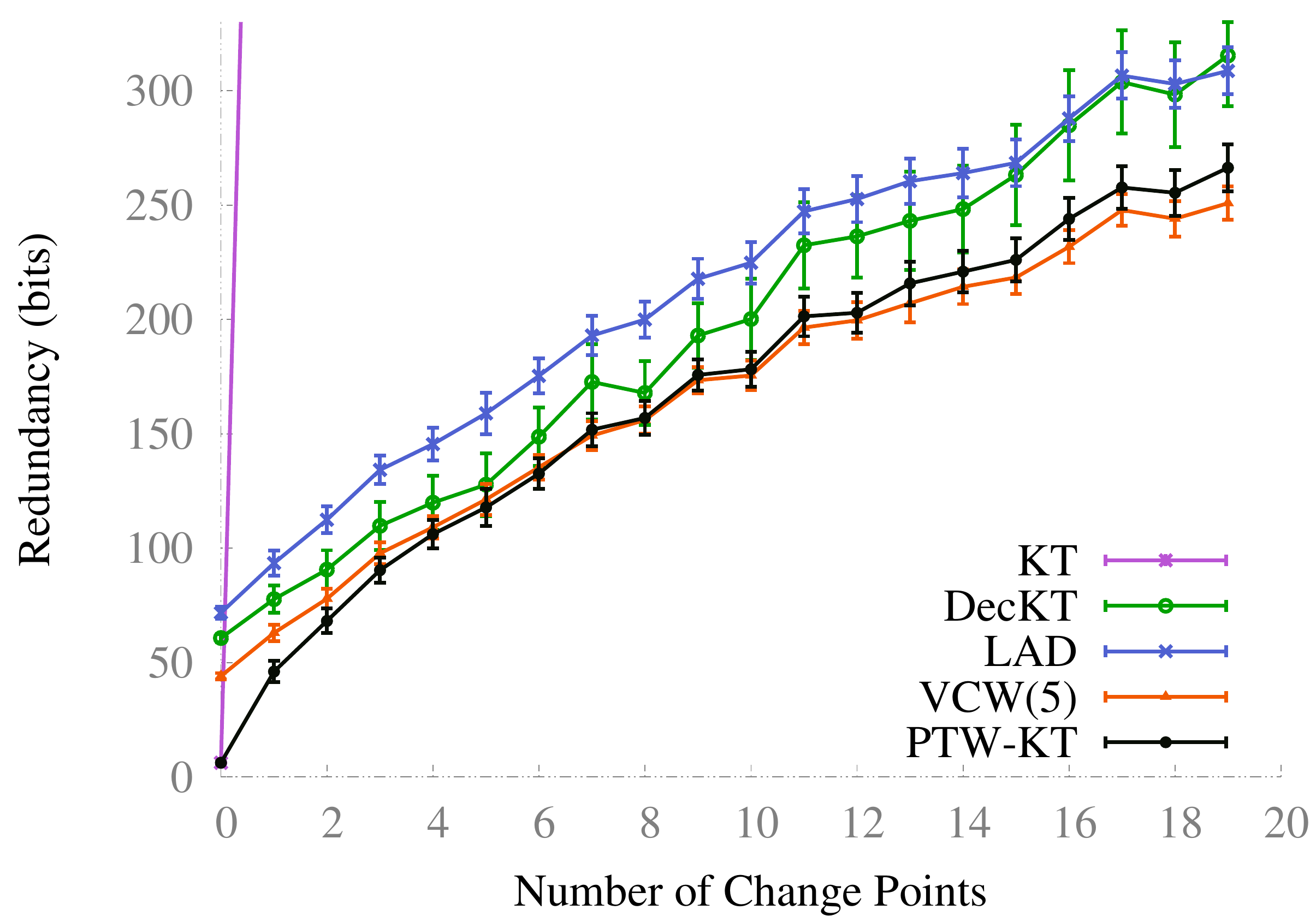}\label{fig:manysplit1}}
   \subfigure[$n = 65536$ over splits $\{0,1, 2, \ldots, 40\}$]{\includegraphics[width=0.495\textwidth]{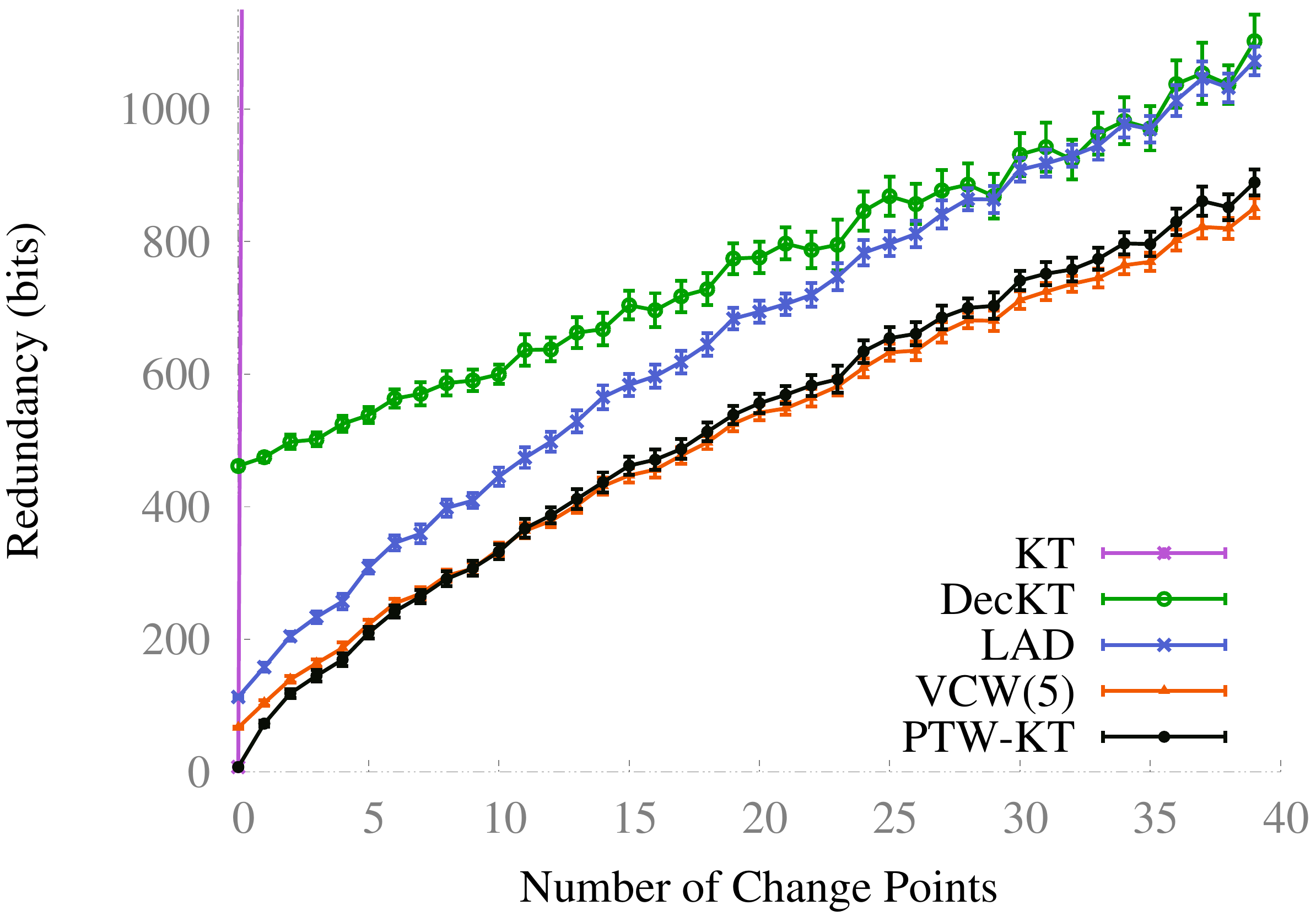}\label{fig:manysplit2}}
  \vspace{-0.5cm}
  \caption{\small Average redundancy of various estimators on binary data for increasing number of change points.}
  \label{fig:manysplits}
\end{figure*}

\begin{table}[!ht]
\centering
\begin{footnotesize}
\begin{tabular}{ | c | c | c | c | c | c | c | c | c | c | c | c | c | c | c | c | c | c | c |}
\hline
\cts~+ & \ss{bib} & \ss{book1} & \ss{book2} & \ss{geo} & \ss{news} & \ss{obj1} & \ss{obj2} & \ss{paper1} & \ss{paper2} & \ss{paper3} & \ss{paper4} & \ss{paper5} & \ss{paper6} & \ss{pic} & \ss{progc} & \ss{progl} & \ss{progp} & \ss{trans} \\ 
\hline\hline
$\deckt$ &1.78 & 2.18 & 1.88 & 4.30 & \bf{2.31} & 3.67 & 2.31 & \bf{2.25}	& \bf{2.20} & \bf{2.45} & \bf{2.75} & \bf{2.87} & \bf{2.33} & 0.78 & \bf{2.28} & 1.60 & 1.62	& 1.36\\
$\kt$ & 1.79 & \bf{2.17} & 1.89 & 4.38 & 2.32 & 3.73 & 2.39 & \bf{2.25}	& \bf{2.20} & \bf{2.45} & 2.76 & 2.89 & 2.34 & 0.79 & 2.30 & 1.61 & 1.64	& 1.37 \\
$\lad$ &2.70 & 2.60 & 2.46 & 4.37 & 3.14 & 4.52 & 3.07 & 3.33	& 3.03 &3.41 & 3.87 & 4.04 & 3.45 & 0.80 & 3.42 & 2.49 & 2.65 & 2.66 \\
$\ladg(5)$ & 2.32 & 2.41 & 2.19 & 4.27 & 2.80 & 4.13 & 2.72 & 2.90 & 2.69 & 3.02 & 3.43 & 3.56 & 3.00 & 0.78 & 2.95 & 2.11 & 2.22 & 2.16 \\
$\ptw\text{-}\kt$ & \bf{1.77} & \bf{2.17} &\bf{1.87} & \bf{4.20} &\bf{2.31} & \bf{3.64} & \bf{2.25} & \bf{2.25}&\bf{2.20} &\bf{2.45} & \bf{2.75} & 2.88  & \bf{2.33} & \bf{0.77} & 2.29 & \bf{1.59} & \bf{1.61} & \bf{1.35}\\
\hline
\end{tabular}
\end{footnotesize}
\vspace{-0.7em}
\caption{\small{Performance (average bits per byte) on the Calgary Corpus}}
\label{tbl:calgary}
\vspace{-0.2em}
\end{table}

Additionally, we evaluated the same set of techniques as a replacement to the KT estimator for the memoryless model used within Context Tree Switching (\cts) \citep{cts_dcc}, a recently introduced universal data compression algorithm for binary, stationary Markov sources of bounded memory.
Performance was measured on the well known Calgary Corpus \citep{Arnold97acorpus}.
Each result was generated using \cts\ with a context depth of 48 bits.
The results (in average bits per byte) are shown in Table \ref{tbl:calgary}. 
Here we see that \ptw-\kt consistently matches or outperforms the other methods.
The largest relative improvements are seen on the non-text files, {\sc geo}, {\sc obj1}, {\sc obj2} and {\sc pic}.
While the performance of \ladg\ could be improved by using a $g > 5$, it was already considerably slower than the other methods.

\paragraph{Tracking.}

\ptw\ can also be used to derive an alternate algorithm for tracking \citep{herbster1998} using the code-length loss.
Consider a base model $\rho$ that is a convex combination of a finite set $\cM~:=~\{ \nu_1, \nu_2, \dots, \nu_{|\cM|} \}$ of $k$-bounded memory models, that is,
\begin{equation}\label{eq:tracking_base_model} 
\rho(x_{1:n} \,|\, x_{1-k:0}) := \sum_{\nu_i \in \cM} w_{\nu_i} \nu_i(x_{1:n} \,|\, x_{1-k:0}),
\end{equation}
where $x_{1-k:0}\in\cX^k$ denotes the initial (possibly empty) context, each $\nu_i$ is a $k$-bounded memory probabilistic data generating source 
(that is, $\nu_i(x_t|x_{<t})=\nu_i(x_t|x_{t-k:t-1})$ for any $t$), $w_\nu \in \mathbb{R}$ and $w_\nu > 0$ for all $\nu \in \cM$, and $\sum_{\nu \in \cM} w_\nu~=~1$. 
We now show that applying \ptw\ to $\rho$ gives rise to a model that will perform well with respect to an interesting subset of the class of switching models.
A switching model is composed of two parts, a set of models $\cM$ and an index set.
An index set $i_{1:n}$ with respect to $\cM$ is an element of $\{ 1, 2, \dots, |\cM| \}^n$.
Furthermore, an index set $i_{1:n}$ can be naturally mapped to a temporal partition in $\cT_n$ by processing the index set sequentially, adding a new segment whenever $i_t \neq i_{t+1}$ for $1 \leq t < n$.
For example, if $|\cM| \geq 2$, the string $1111122221$ maps to the temporal partition $\{ (1,5), (6,9), (10,10) \}$.
The partition induced by this mapping will be denoted by $\cS(i_{1:n})$.
A switching model can then be defined as 
$$\xi_{i_{1:n}}(x_{1:n}) := \hspace{-0.65em} \prod_{(a,b)\in\cS(i_{1:n})} \hspace{-0.65em} \nu_{i_a}(x_{a:b} \,|\, x_{a-k:a-1}),$$
where we have adopted the convention that the previous symbols at each segment boundary define the initializing context for the next bounded memory source.\footnote{This is a choice of convenience. One could always relax this assumption and naively encode the first $k$ symbols of any segment using a uniform probability model, incurring a startup cost of $k \log |\cX|$ bits before applying the relevant bounded memory model. This would increase the upper bound in Equation~\ref{eq:tracking_redundancy} by $\left|\cS(i_{1:n})\right| ( \lceil \log n \rceil + 1) k \log |\cX|$.}
The set of all possible switching models for a sequence of length $n$ with respect to the model class $\cM$ will be denoted by $\cI_n(\cM)$.
If we now let $\tau(x_{1:n})$ denote $\ptw_{\lceil \log n\rceil}(x_{1:n})$ using a base model as defined by Equation \ref{eq:tracking_base_model}, we can state the following upper bound on the redundancy of $\tau$ with respect to an arbitrary switching model.
\begin{corollary}\label{thm:tracking} 
For all $n\in\mathbb{N}$, for any $x_{1:n} \in \cX^n$ and for any switching model $\xi_{i_{1:n}} \in \cI_n(\cM)$, we have
\begin{equation}\label{eq:tracking_redundancy}
-\log \tau(x_{1:n}) + \log \xi_{i_{1:n}}(x_{1:n}) \leq (2 + \kappa) \left|\cS(i_{1:n})\right| ( \lceil \log n \rceil + 1),
\end{equation}
where $\kappa := \max_{\nu \in \cM} -\log(w_{\nu})$.
\ifdcc
\else
\begin{proof}
Using 
$-\log \rho(x_{1:t} \,|\, x_{1-k:0}) = -\log \left\{ \sum_{\nu_i\in\cM} w_{\nu_i}\nu_i(x_{1:t} \,|\, x_{1-k:0}) \right\} \leq -\log w_{\nu^*} -\log \nu^*(x_{1:t} \,|\, x_{1-k:0})$
for any $\nu^* \in \cM$ and $t\in\mathbb{N}$, we see that the redundancy of $\rho$ with respect to any single model in $\cM$ is bounded by $\kappa:=\max_{\nu \in \cM} -\log(w_{\nu})$.
Combining this with Theorem \ref{thm:ptw_redundancy} completes the proof.
\end{proof}
\fi
\end{corollary}
Inspecting Corollary \ref{thm:tracking}, we see that there is a linear dependence on the number of change points and a logarithmic dependence on the sequence length.
Thus we can expect our tracking technique to perform well provided the data generating source can be well modeled by some switching model that changes infrequently.
The main difference between our method and \cite{herbster1998} is that our prior depends on additional structure within the index sequence.
While both methods have a strong prior bias towards favoring a smaller number of change points, the \ptw\ prior arguably does a better job of ensuring that the change points are not clustered too tightly together.
This benefit does however require logarithmically more time and space.

\ifdcc
\else
\section{Extensions and Future Work}

Along the lines of \citep{gyorgy2011efficient}, our method could be extended to a more general class of loss functions within an online learning framework.
We also remark that a technique similar to Context Tree Maximizing \cite{Volf94contextmaximizing} can be applied to \ptw\ to extract the best, in terms of Minimum Description Length \cite{Grunwald2007}, binary temporal partition for a given sequence of data.
Unfortunately we could not find a way to avoid building a full context tree for this case, which means that $O(n \log n)$ memory would be required instead of the $O(\log n)$ required by Algorithm \ref{alg:ptw}.
Another interesting follow up would be to generalize Theorem 3 in \citep{cts_dcc} to the piecewise stationary setting. 
Combining such a result with Corollary \ref{thm:ptw_kt} would allow us to derive a redundancy bound for the algorithm that uses \ptw-\kt\ for the memoryless model within \cts; this bound would hold with respect to any binary, piecewise stationary Markov source of bounded memory.
\fi

\section{Conclusion}

This paper has introduced Partition Tree Weighting, an efficient meta-algorithm that automatically generalizes existing coding distributions to their piecewise stationary extensions.
Our main contribution is to introduce a prior, closely related to the Context Tree Weighting method, to efficiently weight over a large subset of possible temporal partitions.
The order of the redundancy and the complexity of our algorithm matches those of the best competitors available in the literature, with the new algorithm exhibiting a superior complexity-performance trade-off in our experiments. 

\paragraph{Acknowledgments.}
The authors would like to thank Marcus Hutter for some helpful comments.
This research was supported by NSERC and Alberta Innovates Technology Futures.

{
\ifdcc
\else
\small
\fi
\bibliographystyle{plainnat} 
\bibliography{ptw}
}

\ifdcc
\else

\newpage
\appendix

\appsec{Supplementary Proofs}
\label{sec:proofs}
\vspace{2em}

\lemmaPtwLemma*
\begin{proof}
This is a straightforward adaptation of Lemma 2 from \citep{ctw95}.
We use induction on $d$.
First note that $\ptw_0(x_{1:n}) = \rho(x_{1:n})$ by definition, so the base case of $d=1$ holds trivially.
Now assume that Equation \ref{eq:ptw_lemma} holds for some depth $d-1$, and observe that
\begin{eqnarray*}
\ptw_d(x_{1:n}) 
&=& \sum_{\cP \in \cC_d(1)} 2^{-\Gamma_d(\cP)} \prod_{(i,j) \in \cP} \rho(x_{i:j}) \\
&=& \tfrac{1}{2} \rho(x_{1:n}) + \hspace{-0.5em} \sum_{\cP \in \cC_d(1) \setminus \{\{(1,2^d)\}\}} \hspace{-1.75em} 2^{-\Gamma_d(\cP)} \prod_{(i,j) \in \cP} \rho(x_{i:j}) \\
&=& \tfrac{1}{2} \rho(x_{1:n}) + \tfrac{1}{2} \hspace{-0.35em} \sum_{\substack{\cP_1 \in \cC_{d-1}(1)\\ \cP_2 \in \cC_{d-1}(k+1)} } \hspace{-1.7em} 2^{-\Gamma_{d-1}(\cP_1)-\Gamma_{d-1}(\cP_2)} \hspace{-0.5em} \prod_{(i,j) \in \cP_1} \hspace{-0.5em} \rho(x_{i:j}) \hspace{-0.5em} \hspace{-0.25em} \prod_{(r,s) \in \cP_2} \hspace{-0.5em} \rho(x_{r:s}) \\
&=& \tfrac{1}{2} \rho(x_{1:n}) +  \tfrac{1}{2} \hspace{-0em}  \left ( \sum_{\cP_1 \in \cC_{d-1}(1)} \hspace{-1.35em} 2^{-\Gamma_{d-1}(\cP_1)} \hspace{-0.3em} \prod_{(i,j) \in \cP_1} \hspace{-0.6em} \rho(x_{i:j}) \hspace{-0em} \right )\hspace{-0.45em} \left ( \sum_{\cP_2 \in \cC_{d-1}(k+1)} \hspace{-1.35em} 2^{-\Gamma_{d-1}(\cP_2)} \hspace{-0.75em} \prod_{(r,s) \in \cP_2} \hspace{-0.75em} \rho(x_{r:s}) \right ) \\
&=& \frac{1}{2} \rho(x_{1:n}) + \frac{1}{2} \ptw_{d-1} \left( x_{1:k} \right) \ptw_{d-1}\left(x_{k+1:n} \right)
\end{eqnarray*}
where $k := 2^{d-1}$.
The third step uses the property 
$\Gamma_{d}(\cP_1 \cup \cP_2) = \Gamma_{d-1}(\cP_1) + \Gamma_{d-1}(\cP_2) + 1$, 
which holds when $\cP_1 \in \cC_{d-1}(1)$ and $\cP_2 \in \cC_{d-1}(k+1)$.  
The final step applies the inductive hypothesis.
\end{proof}

\propStronglyOnline*
\begin{proof}
We begin by showing that for any $t \le 2^k$ we have
\begin{equation}
\label{eq:ptwratio}
\tfrac{2}{3}  \ptw_{k+1}(x_{1:t}) \leq \ptw_{k}(x_{1:t}).
\end{equation}
Using Lemma \ref{lem:ptw_lemma}, 
\begin{equation}\label{eq:ptw_lemma_twist}
\ptw_{k+1}(x_{1:t}) 
= \frac{1}{2} \rho(x_{1:t}) + \frac{1}{2} \ptw_{k} \left( x_{1:2^k} \right) \ptw_{k}\left(x_{2^k+1:2^{k+1}} \right) 
= \frac{1}{2} \rho(x_{1:t}) + \frac{1}{2} \ptw_{k} \left( x_{1:t} \right)
\end{equation}
holds for $t \leq 2^k$.
On the other hand,
\begin{equation*}
\ptw_k(x_{1:t}) 
= \sum_{\cP\in \cC_k} 2^{-\Gamma_k(\cP)} \hspace{-0.4em} \prod_{(a,b)\in\cP} \hspace{-0.4em} \rho(x_{a:b})
= \frac{1}{2}\rho(x_{1:t}) + \hspace{-1.5em} \sum_{\cP\in \cC_k \setminus \{\{ (1,2^k) \}\} } \hspace{-1.5em} 2^{-\Gamma_k(\cP)} \hspace{-0.4em} \prod_{(a,b)\in\cP} \rho(x_{a:b}) 
\geq \tfrac{1}{2}\rho(x_{1:t}),
\end{equation*}
which, together with Equation \ref{eq:ptw_lemma_twist} proves Equation
\ref{eq:ptwratio}. The latter allows us to derive a lower bound on $\ptw(x_{1:n})$, by noting that
\begin{eqnarray*}
\ptw(x_{1:n}) &=& \prod_{i=1}^n \ptw_{\lceil \log i \rceil}(x_{i} \cdbar x_{<i}) \\
&=& \ptw_0(x_1) \left(\prod_{a=1}^{d-1} \ptw_a(x_{2^{a-1}+1:2^a} \cdbar x_{1:2^{a-1}}) \right) \ptw_d(x_{2^{d-1}+1:n} \cdbar x_{1:2^{d-1}}) \\
&=& \ptw_0(x_1) \left(\prod_{a=1}^{d-1} \frac{\ptw_a(x_{1:2^a})}{\ptw_a(x_{1:2^{a-1}})} \right) \frac{\ptw_d(x_{1:n})}{\ptw_d(x_{1:2^{d-1}})} \\
&=& \ptw_d(x_{1:n}) \prod_{a=1}^{d} \frac{\ptw_{a-1}(x_{1:2^{a-1}})}{\ptw_a(x_{1:2^{a-1}})} \\
&\ge& \left( \tfrac{2}{3} \right)^d \ptw_d(x_{1:n}).
\end{eqnarray*}
Hence, taking the negative logarithm of both sides we obtain
\[
-\log \ptw(x_{1:n}) 
\le -\log \ptw_d(x_{1:n}) + d \log(\tfrac{3}{2}) 
= -\log \ptw_d(x_{1:n}) + \lceil \log n \rceil (\log 3 -1).
\]
\end{proof}

\fi

\end{document}